\documentclass[11pt,reqno]{amsart}
\usepackage{amsfonts}
\usepackage{mathrsfs}
\allowdisplaybreaks[4]
\usepackage{graphicx} %We can use any other package if it is necessary
\usepackage{epsfig}
\usepackage{amssymb}
\usepackage{amsmath}
\usepackage{cite}
\newtheorem{theorem}{Theorem}
\newtheorem{definition}{Definition}
\newtheorem{proposition}[theorem]{Proposition}

\newtheorem{lemma}[theorem]{Lemma}

\newcommand{\pa}{\partial}
\newcommand{\J}{\mathcal{J}}

\newcommand{\K}{\mathcal{K}}
\renewcommand{\S}{\mathcal{S}}
\newcommand{\W}{\mathcal{W}}
\renewcommand{\L}{\mathcal{L}}
\newcommand\Zop{\mathbb{Z^{\mathrm{odd}}_+}}

\newcommand\B{\mathbb{B}}
\newcommand\Z{\mathbb{Z}}
\newcommand\M{\mathcal{M}}

\def\[{\begin{eqnarray}}
\def\]{\end{eqnarray}}
\def\d{\partial}

\begin{document}
\title{Quantum torus algebras and B(C) type Toda systems}
\author{Na Wang\dag, Chuanzhong Li\ddag  }
\allowdisplaybreaks
\dedicatory {\small
\dag Department of Mathematics and Statistics, Henan University, Kaifeng, 475001, China\\
\ddag Department of Mathematics,  Ningbo University, Ningbo, 315211, China\\
}

\thanks{\ddag Corresponding author's email:lichuanzhong@nbu.edu.cn}
\begin{abstract}
In this paper, we construct a new even constrained B(C) type Toda hierarchy and derive its B(C) type Block type additional symmetry. Also we generalize the B(C) type Toda hierarchy to the $N$-component B(C) type Toda hierarchy which is proved to have symmetries of a coupled $\bigotimes^NQT_+ $ algebra
( $N$-folds direct product of the
positive half of the quantum torus algebra $QT$).

\end{abstract}
\maketitle

\tableofcontents
%%%%%%%%%%%%%%%%%%%%%%%%%%%%%%%%%%%%%%%%%%%%%%%%%%%%%%
\section{Introduction}
%%%%%%%%%%%%%%%%%%%%%%%%%%%%%%%%%%%%%%%%%%%%%%%%%%%%%%%
The Toda lattice hierarchy
  as  a completely integrable system  has many important applications in mathematics and physics including the representation theory of Lie algebras and  random
matrix models  \cite{Toda,UT,witten}. The Toda system has many kinds of reductions or extensions, for example the B and C type Toda hierarchies\cite{UT,cheng2011}, extended Toda hierarchy (ETH)\cite{CDZ}, bigraded Toda hierarchy (BTH)\cite{C}-\cite{RMP} and so on. There are some other generalizations called  multi-component Toda systems\cite{UT,manasInverse2} which are useful in the fields of multiple orthogonal
polynomials and non-intersecting Brownian motions.

The multicomponent 2D Toda hierarchy
 was considered from the point of view of the Gauss-Borel factorization problem, the theory  of multiple matrix orthogonal polynomials, non-intersecting Brownian motions and matrix Riemann-Hilbert problem \cite{manasInverse2}-\cite{manas}. In fact the multicomponent 2D Toda hierarchy in \cite{manasinverse} is a periodic reduction of the bi-infinite matrix-formed two dimensional Toda hierarchy.
In \cite{EMTH}, we generalize the multicomponent Toda hierarchy to an extended multicomponent Toda hierarchy including extended logarithmic flow equations. Later by a commutative algebraic reduction on the extended multicomponent Toda hierarchy, we get an extended $Z_N$-Toda hierarchy\cite{EZTH} which might be useful in Gromov-Witten theory.

This paper is organized in the following way. In Section 2, we
recall some basic knowledge about the B(C) type Toda hierarchy. We construct a new even constrained B(C) type Toda hierarchy and derive its Block type additional symmetry
 in Section 3. Next, in Section 4 we generalize the B(C) type Toda hierarchy to a new $N$-component B(C) type Toda hierarchy. In the last section, we construct the symmetry of the $N$-component B(C) type Toda hierarchy which constitutes a coupled $\bigotimes^NQT_+ $ algebra
( $N$-folds direct product of the
positive half of the quantum torus algebra $QT$).

%%%%%%%%%%%%%%%%%%%%%%%%%%%%%%%%%%%%%%%%%%%%%%%%%%%%%%
\section{The B(C) type Toda  hierarchy}
%%%%%%%%%%%%%%%%%%%%%%%%%%%%%%%%%%%%%%%%%%%%%%%%%%%%%%%
In this section, some basic facts about the B(C) type Toda  hierarchy
are reviewed. One can refer to \cite{UT,cheng2011} for
 more details about the B(C) type Toda  hierarchy (or BTH(CTH)).

Then the BTH  hierarchy is defined in the Lax forms as
\begin{equation}\label{bctlhierarchy}
    \pa_{x_{2n+1}}L_1=[-(L^{2n+1}_1)_-,L_1]\ \ \ \textmd{and}\ \ \ \pa_{y_{2n+1}}L_1=[-(L_2^{2n+1})_-,L_1],\ \ \ n=0,1,2,\cdots,
\end{equation}
\begin{equation}\label{bctlhierarchy}
    \pa_{x_{2n+1}}L_2=[(L^{2n+1}_1)_+,L_2]\ \ \ \textmd{and}\ \ \ \pa_{y_{2n+1}}L_2=[(L_2^{2n+1})_+,L_2],\ \ \ n=0,1,2,\cdots,
\end{equation}
where the Lax operator $L_i$ is given by  a pair of infinite matrices
\begin{equation}\label{laxoperator}
    L_1=\sum _{-\infty<i\leq 1}\textmd{diag}[a_i^{(1)}(s)]\Lambda^i,\ \
    L_2=\sum _{-1\leq i<\infty}\textmd{diag}[a_i^{(2)}(s)]\Lambda^i,
\end{equation}
with $\Lambda=(\delta_{j-i,1})_{i,j\in \mathbb{Z}}$, and
$a_i^{(k)}(s)$ and  $a_i^{(k)}(s)$ depending on
$x=(x_1,x_3,x_5,\cdots)$ and $y=(y_1,y_3,y_5,\cdots)$, such that
$$a_1^{(1)}(s)=1 \ \ \ \text{and}\ \ \ a_{-1}^{(2)}(s)\neq 0\ \ \ \forall s$$
and satisfies the BTH(CTH)  constraint\cite{UT}
\begin{equation}\label{bctlconstr}
    L_i^T=-JL_iJ^{-1},\ \ L_i^T=-KL_iK^{-1},
\end{equation}
where $J=((-1)^i\delta_{i+j,0})_{i,j\in\mathbb{Z}}, K=\Lambda J$
and $T$ refers to the matrix transpose. The BTH  constraint
 is explicitly
showed as
\begin{equation}\label{bctlconstrcomponent}
    a_i^{(k)}(s)=(-1)^{i+1}a_i^{(k)}(-s-i)\ , k=1,2.
\end{equation}
The CTH  constraint
 means
\[ a_i^{(k)}(s)=(-1)^{i+1}a_i^{(k)}(-s-i-1).\]

The Lax equation for the BTH(CTH) can be expressed as
a system of equations of the Zakharov-Shabat type:
\begin{eqnarray}
&&\pa_{x_{2n+1}}(L^{2m+1}_1)_+-\pa_{x_{2m+1}}(L^{2n+1}_1)_++[(L^{2m+1}_1)_+,(L^{2n+1}_1)_+]=0,\label{zs1}\\
&&\pa_{y_{2n+1}}(L^{2m+1}_2)_-+\pa_{y_{2m+1}}(L^{2n+1}_2)_--[(L^{2m+1}_2)_-,(L^{2n+1}_2)_-]=0,\label{zs2}\\
&&\pa_{y_{2n+1}}(L^{2m+1}_1)_++\pa_{x_{2m+1}}(L^{2n+1}_2)_--[(L^{2m+1}_1)_+,(L^{2n+1}_2)_-]=0,\label{zs2}\\
&&-\pa_{y_{2n+1}}(L^{2m+1}_2)_--\pa_{x_{2m+1}}(L^{2n+1}_1)_+-[(L^{2m+1}_2)_-,(L^{2n+1}_1)_+]=0.
\end{eqnarray}
When $m=n=0$, one can get the B type Toda  equation
\begin{eqnarray}
&&\pa_{x_1}a_{-1}^{(2)}(1)=a_{-1}^{(2)}(1)a_{0}^{(1)}(1),\ \
\pa_{x_1}a_{-1}^{(2)}(s)=a_{-1}^{(2)}(s)(a_{0}^{(1)}(s)-a_{0}^{(1)}(s-1))\
(s\geq
2),\nonumber\\
&&\pa_{y_1}a_{0}^{(1)}(s)=a_{-1}^{(2)}(s)-a_{-1}^{(2)}(s+1)\
(s\geq1),\label{BTHequation}
\end{eqnarray}
by considering the corresponding constraint
(\ref{bctlconstrcomponent}).
Also one can get the  C type Toda  equation
\begin{eqnarray}
&&\pa_{x_1}a_{-1}^{(2)}(0)=2a_{-1}^{(2)}(0)a_{0}^{(1)}(0),\ \
\pa_{x_1}a_{-1}^{(2)}(s)=a_{-1}^{(2)}(s)(a_{0}^{(1)}(s)-a_{0}^{(1)}(s-1))\
(s\geq
1),\nonumber\\
&&\pa_{y_1}a_{0}^{(1)}(s)=a_{-1}^{(2)}(s)-a_{-1}^{(2)}(s+1)\
(s\geq0),\label{ctlequation}
\end{eqnarray}
The Lax operator of the BTH(CTH) (\ref{bctlhierarchy})
has the representation
\[
L_1&=&W_1\Lambda W_1^{-1}=S_1\Lambda S_1^{-1},\\
L_2&=&W_2\Lambda^{-1}W_2^{-1}=S_2\Lambda^{-1}S_2^{-1},
\] where
\begin{eqnarray}
S_1(x,y)=\sum_{i\geq 0} \textmd{diag}[c_i(s;x,y)] \Lambda^{-i},&&S_2(x,y)=\sum_{i\geq 0} \textmd{diag}[c_i'(s;x,y)] \Lambda^{i} \label{smatrices}
\end{eqnarray}
and
\begin{equation}\label{wmatrices}
    W_1(x,y)=S_1(x,y)e^{\xi(x,\Lambda)},\quad W_2(x,y)=S_2(x,y)e^{\xi(y,\Lambda^{-1})}
\end{equation}
with $c_0(s;x,y)=1$ and $c_o'(s;x,y)\neq 0$ for any $s$, and $\xi(x,\Lambda^{\pm1})=\sum_{n\geq0}x_{2n+1}\Lambda^{\pm 2n+1}$.  For the B type Toda hierarchy, under an appropriate choice $(W_1,W_2)$ satisfies
\begin{equation}\label{bcwconstraints}
    J^{-1}W_i^TJ=W_i^{-1}\ , i=1,2.
\end{equation}
 For the C type Toda hierarchy, under an appropriate choice $(W_1,W_2)$ satisfies
\begin{equation}\label{bcwconstraints}
    K^{-1}W_i^TK=W_i^{-1}\ , i=1,2.
\end{equation}
The wave operators evolve as
\begin{eqnarray}
\pa_{x_{2n+1}}S_1&=&-(L_1^{2n+1})_-S_1,\quad \pa_{y_{2n+1}}S_1=-(L_2^{2n+1})_-S_1,\label{sevolution}\\
\pa_{x_{2n+1}}W_1&=&(L_1^{2n+1})_+W_1,\quad \pa_{y_{2n+1}}W_1=-(L_2^{2n+1})_-W_2,\label{wevolution}
\end{eqnarray}
\begin{eqnarray}
\pa_{x_{2n+1}}S_2&=&(L_1^{2n+1})_+S_2,\quad \pa_{y_{2n+1}}S_2=(L_2^{2n+1})_+S_2,\label{sevolution}\\
\pa_{x_{2n+1}}W_2&=&(L_1^{2n+1})_+W_2,\quad \pa_{y_{2n+1}}W_2=-(L_2^{2n+1})_-W_2.\label{wevolution}
\end{eqnarray}

 At last, we end this section with the introduction of the additional symmetries of the BTH(CTH). The Orlov-Shulman operator\cite{cheng2011} is defined
as
\begin{equation}\label{osoperator}
    M_1=W_1\varepsilon W_1^{-1},\ \ M_2=W_2\varepsilon^*W_2^{-1},
\end{equation}
where
$$\varepsilon=\rm{diag}[s]\Lambda^{-1},\quad \varepsilon^*=-J\varepsilon J^{-1},$$
 satisfying
\begin{eqnarray}
&&[L_i,M_i]=1, \nonumber\\
&&\pa_{x_{2n+1}}M_i=[(L^{2n+1}_1)_+,M_i],
\pa_{y_{2n+1}}M_i=[-(L_2^{2n+1})_-,M_i]. \label{osproperty}
\end{eqnarray}
To construct the Block symmetry of the BTH, the following lemma should be introduced.
\begin{lemma}\label{MBTHlemma}
The following identities hold true
\begin{eqnarray}
 \Lambda^{-1}\varepsilon  \Lambda=J^{-1}\varepsilon^T J,&&\Lambda\varepsilon^*  \Lambda^{-1}=J^{-1}\varepsilon^{*T} J,\label{blemma}
\end{eqnarray}
\begin{eqnarray}
\varepsilon=K\varepsilon^T K^{-1},&&\varepsilon^* =K\varepsilon^{*T} K^{-1}.\label{clemma}
\end{eqnarray}
\end{lemma}
For the BTH, using the above lemma, one can derive
\begin{eqnarray}
{M_i}^T&=&J{L_i}^{-1}{M_i}{L_i}J^{-1}.\label{bmtransposeB}
\end{eqnarray}
For the CTH, using the above lemma, one can derive
\begin{eqnarray}
{M_i}^T&=&K{M_i}K^{-1}.\label{bmtransposeC}
\end{eqnarray}
The additional symmetry \cite{cheng2011} of the BTH can be defined by
introducing the additional independent variables $x_{m,l}$ and
$y_{m,l}$,
\begin{equation}\label{bctladdsym}
    \pa_{x_{m,l}}W_1=-A_{ml}(M_1,L_1)_-W_1,\quad \pa_{y_{m,l}}W_1=-A_{ml}(M_2,L_2)_-W_1,
\end{equation}
\begin{equation}\label{bctladdsym}
    \pa_{x_{m,l}}W_2=A_{ml}(M_1,L_1)_+W_2,\quad \pa_{y_{m,l}}W_2=A_{ml}(M_2,L_2)_+W_2,
\end{equation}
where
\begin{equation}\label{aBTHaddsymm}
A_{ml}({M_i},{L_i})={M_i}^m{L_i}^l-(-1)^l{L_i}^{l-1}{M_i}^m{L_i}.
\end{equation}
For the case of the CTH, the operator $A_{ml}$ will become
\begin{equation}\label{aCTHaddsymm}
A_{ml}({M_i},{L_i})={M_i}^m{L_i}^l-(-1)^l{L_i}^{l}{M_i}^m.
\end{equation}
These additional flows form a coupled $W_{\infty}$ Lie algebra \cite{cheng2011}.

\section{The even constrained BTH(CTH)}
%%%%%%%%%%%%%%%%%%%%%%%%%%%%%%%%%%%%%%%%%%%%%%%%%%%%%%%
In this section, for a new constrained BTH(CTH),
 the Lax operator $L$ is given by  an infinite matrices $L$ as
\begin{equation}\label{laxoperator}
    L=L_1^{2N}=L_2^{2M}=\sum _{-2M<i\leq 2N}\textmd{diag}[a_i(s)]\Lambda^i,
\end{equation}
with
$a_{2N}(s)=1 ,$
and for the BTH, it satisfies the B type  constraint
\begin{equation}
    L^T=JLJ^{-1},\
\end{equation}
and for the CTH, it satisfies the C type  constraint
\begin{equation} L^T=KLK^{-1}.
\end{equation}
Then the constrained  BTH(CTH)  hierarchy is defined in the Lax forms as
\begin{equation}\label{bctlhierarchy}
    \pa_{x_{2n+1}}L=[-(L^{\frac{2n+1}{2N}})_-,L]\ \ \ \textmd{and}\ \ \ \pa_{y_{2n+1}}L=[-(L^{\frac{2n+1}{2N}})_-,L],\ \
\end{equation}
\begin{equation}\label{bctlhierarchy}
    \pa_{x_{2n+1}}L=[(L^{\frac{2n+1}{2M}})_+,L]\ \ \ \textmd{and}\ \ \ \pa_{y_{2n+1}}L=[(L^{\frac{2n+1}{2M}})_+,L],\ \ \ n=0,1,2.\cdots
\end{equation}

The Lax operator of the constrained BTH(CTH) (\ref{bctlhierarchy})
has the representation
\[
L&=&W_1\Lambda^{2N} W_1^{-1}=W_2\Lambda^{-2M}W_2^{-1},
\] where
\begin{eqnarray}
S_1(x,y)=\sum_{i\geq 0} \textmd{diag}[c_i(s;x,y)] \Lambda^{-i},&&S_2(x,y)=\sum_{i\geq 0} \textmd{diag}[c_i'(s;x,y)] \Lambda^{i} \label{smatrices}
\end{eqnarray}
and
\begin{equation}\label{wmatrices}
    W_1(x,y)=S_1(x,y)e^{\xi(x,\Lambda)},\quad W_2(x,y)=S_2(x,y)e^{\xi(y,\Lambda^{-1})}
\end{equation}
with $c_0(s;x,y)=1$ and $c_o'(s;x,y)\neq 0$ for any $s$, and $\xi(x,\Lambda^{\pm1})=\sum_{n\geq0}x_{2n+1}\Lambda^{\pm 2n+1}$.  Under an appropriate choice $(W_1,W_2)$ of the constrained BTH(CTH) satisfies
\begin{equation}\label{bcwconstraints}
    J^{-1}W_i^TJ=W_i^{-1}, (K^{-1}W_i^TK=W_i^{-1}) , i=1,2.
\end{equation}
The wave operators evolve according to
\begin{eqnarray}
\pa_{x_{2n+1}}S_1&=&-(L^{\frac{2n+1}{2N}})_-S_1,\quad \pa_{y_{2n+1}}S_1=-(L^{\frac{2n+1}{2M}})_-S_1,\label{sevolution}\\
\pa_{x_{2n+1}}W_1&=&(L^{\frac{2n+1}{2N}})_+W_1,\quad \pa_{y_{2n+1}}W_1=-(L^{\frac{2n+1}{2M}})_-W_1,\label{wevolution}
\end{eqnarray}
\begin{eqnarray}
\pa_{x_{2n+1}}S_2&=&(L^{\frac{2n+1}{2N}})_+S_2,\quad \pa_{y_{2n+1}}S_2=(L^{\frac{2n+1}{2M}})_+S_2,\label{sevolution}\\
\pa_{x_{2n+1}}W_2&=&(L^{\frac{2n+1}{2N}})_+W_2,\quad \pa_{y_{2n+1}}W_2=-(L^{\frac{2n+1}{2M}})_-W_2.\label{wevolution}
\end{eqnarray}

The Orlov-Shulman operator $\bar M_i$ will be defined as
as
\begin{equation}\label{osoperator}
   \bar M_1=W_1\varepsilon_{2N} W_1^{-1},\ \ \bar M_2=W_2\varepsilon_{-2M}^*W_2^{-1},
\end{equation}
where
$$\varepsilon_{2N}=\frac{1}{2N}\rm{diag}[s]\Lambda^{-2N},\quad \varepsilon_{-2M}^*=-\frac{1}{2M}\varepsilon^T\Lambda^{2M},$$
 satisfying
\begin{eqnarray}
&&[L,\bar M_i]=1, \nonumber\\
&&\pa_{x_{2n+1}}\bar M_i=[(L^{\frac{2n+1}{2N}})_+,\bar M_i],
\pa_{y_{2n+1}}\bar M_i=[-(L^{\frac{2n+1}{2M}})_-,\bar M_i]. \label{osproperty}
\end{eqnarray}

\begin{lemma}\label{asymM-M2}
The difference of two Orlov-Schulman operators $\bar M_i$ for constrained BTH hierarchy has following B type property:
\begin{align}\label{leidentity1}
L^T(\bar M_1-\bar M_2)^T=J(L\bar M_1-L\bar M_2)J^{-1},
\end{align}
and for constrained CTH hierarchy has following C type property:
\begin{align}
\label{leidentity2}L^T(\bar M_1-\bar M_2)^T=K(L\bar M_1-L\bar M_2)K^{-1}.
\end{align}
\end{lemma}

\begin{proof}
It is easy to find  the two Orlov-Schulman operators can be expressed as
\[ \label{cMandM}\bar M_1=\frac{M_1L_1^{1-2N}}{2N},\ \ \bar M_2=-\frac{ M_2L_2^{1-2M}}{2M}.
\]

Putting eq.\eqref{cMandM} into $(\bar M_1-\bar M_2)^T$ can lead to
\begin{align}
&(\bar M_1-\bar M_2)^T=\frac{JL_1^{-2N}M_1L_1J^{-1}}{2N}+\frac{JL_2^{-2M}M_2L_2^{-1} J^{-1}}{2M}\\
&=\frac{JL_1^{-2N}M_1L_1J^{-1}}{2N}+\frac{JL_2^{-2M}M_2L_2^{-1} J^{-1}}{2M}\\
&=\frac{J(M_1L_1^{1-2N}-2NL_1^{-2N})J^{-1}}{2N}+\frac{J(M_2L_2^{1-2M}+2ML_2^{-2}) J^{-1}}{2M},
\end{align}
which can further lead to eq.\eqref{leidentity1}.
For the CTH, one can do the similar calculation as
\begin{align}
&(\bar M_1-\bar M_2)^T=\frac{KL_1^{1-2N}M_1K^{-1}}{2N}+\frac{KL_2^{1-2M}M_2 K^{-1}}{2M}\\
&=\frac{KL_1^{1-2N}M_1K^{-1}}{2N}+\frac{KL_2^{1-2M}M_2 K^{-1}}{2M}\\
&=\frac{K(M_1L_1^{1-2N}-2NL_1^{-2N})K^{-1}}{2N}+\frac{K(M_2L_2^{1-2M}+2ML_2^{-2}) K^{-1}}{2M},
\end{align}
which can further lead  to eq.\eqref{leidentity2}

In above calculation, the commutativity between $L$ and $\bar M_1-\bar M_2$ is already used.
Till now, the proof is finished.
\end{proof}

For the constrained BTH(CTH), we need the following operator
\begin{align}
\B_{m,l}=(\bar M_1-\bar M_2)^{m}L^l,\ \  m\in \Zop, l\in \Z_+.
\end{align}
One can easily check that for the BTH
\begin{align}
\B_{m,l}^T=J\B_{m,l}J^{-1},\ \  m\in \Zop,
\end{align}
and for the CTH
\begin{align}
\B_{m,l}^T=K\B_{m,l}K^{-1},\ \  m\in \Zop.
\end{align}
That means it is reasonable to define additional flows of the constrained BTH(CTH) as
\begin{align}\label{blockflow}
&\frac{\partial L}{\partial c_{m,l}}=[-({\B}_{m,l})_-, L],\ \ m\in \Zop, l\in \Z_+.
\end{align}

\begin{proposition}
For the BTH(CTH), the flows \eqref{blockflow} can commute with original flows of  the BTH(CTH), namely,
\begin{equation*}
\left[\frac{\partial}{\partial c_{m,l}}, \frac{\partial}{\partial x_k}\right]=0, \quad
\left[\frac{\partial}{\partial c_{m,l}}, \frac{\partial}{\partial y_k}\right]=0, \qquad l\in \Z_{+}, ~m,k\in\Zop,
\end{equation*}
which hold in the sense of acting on  $W_i$ or $L.$

\end{proposition}

\begin{theorem}\label{thm-MLs}
The flows in eq.\eqref{blockflow} about  additional symmetries of constrained BTH(CTH) compose following Block type Lie algebra

\[[\partial_{c_{m,l}},\partial_{c_{s,k}}]=(km-s l)\partial_{c_{m+s-1,k+l-1}},\ \  m,s\in \Zop, k,l\in \Z_+,\]

which holds in the sense of acting on  $W_i$ or $L.$
\end{theorem}

\section{ Multicomponent B(C) type Toda hierarchy}

In this section we will introduce the multicomponent B type Toda hierarchy (MBTH) and multicomponent C type Toda hierarchy (MCTH).
In the following, we denote $E_{\Z\times \Z}$ as the bi-infinite identity matrix and $E_{N\times N}$ as the $N\times N$ identity matrix.
We also denote $E_{kk}$ as a $N\times N$ matrix which is $1$ at the position of the k-th row and k-th column and $0$ for other elements.
The Lax operators $\L_1,\L_2$ of the MBTH(MCTH) are given by  a pair of infinite matrices
\begin{equation}\label{laxoperatorM}
    \L_1=\sum _{-\infty<i\leq 1}\textmd{diag}[b_i(s)]\bar\Lambda^i,\ \ \L_2=\sum _{-1\leq i<\infty}\textmd{diag}[c_i(s)]\bar\Lambda^i,
\end{equation}
where $b_i(s),c_i(s)$ are matrices of size $N\times N$ and  $\bar\Lambda^i=\Lambda^i\otimes E_{N\times N}$ and they
satisfy the B type(C Type)  constraint\cite{UT}
\begin{equation}\label{bctlconstr}
    \L_i^T=-\J\L_i\J^{-1}(\L_i^T=-\K\L_i\K^{-1}),
\end{equation}
where $\J=((-1)^i\delta_{i+j,0})_{i,j\in\mathbb{Z}}\otimes E_{N\times N}$, $\K=\Lambda J\otimes E_{N\times N}$.
Here the product $\otimes$ is the Kronecker product between a matrix of size $\Z\times \Z$ and a matrix of size $N\times N$. Let us first introduce some convenient notations as $\bar E_{kk}=E_{\Z\times \Z}\otimes E_{kk}$.

The Lax operators of the MBTH(MCTH) (\ref{bctlhierarchy})
can have the following dressing structure
\[
\L_1&=&\W_1\bar \Lambda \W_1^{-1}=\S_1\bar \Lambda \S_1^{-1},\\
\L_2&=&\W_2\bar \Lambda^{-1}\W_2^{-1}=\S_2\bar \Lambda^{-1}\S_2^{-1},
\] where
\begin{equation}\label{wmatrices}
    \W_1(x,y)=\S_1(x,y)(e^{\xi(x,\Lambda)}\otimes E_{N\times N}),\quad \W_2(x,y)=\S_2(x,y)(e^{\xi(y,\Lambda^{-1})}\otimes E_{N\times N}).
\end{equation}
Now we define matrix operators $C_{kk},\bar C_{kk},B_{jk},\bar B_{jk}$ as follows
\begin{align}\label{satoS}
\begin{aligned}
C_{kk}&:=\W_1\bar E_{kk}\W_1^{-1},\ \ \bar C_{kk}:= \W_2 \bar E_{kk} \W_2^{-1},\\
B_{jk}&:=\W_1\bar E_{kk}\bar \Lambda^j\W_1^{-1},\ \ \bar B_{jk}:= \W_2 \bar E_{kk}\bar \Lambda^{-j} \W_2^{-1}.
\end{aligned}
\end{align}
Now we give the definition of the multicomponent B(C) type Toda hierarchy(MBTH).
\begin{definition}The multicomponent B(C) type Toda hierarchy is a hierarchy in which the dressing operators $\S_1,\S_2$ satisfy following Sato equations
\begin{align}
\label{satoSt} \partial_{t_{jk}}\S_1&=-(B_{jk})_-\S_1,& \partial_{t_{jk}} \S_2&=(B_{jk})_+ \S_2,  \\
\label{satoSbart}
\partial_{\bar t_{jk}}\S_1&=-(\bar B_{jk})_- \S_1,& \partial_{\bar t_{jk}} \S_2&=(\bar B_{jk})_+ \S_2.\end{align}
\end{definition}
Then one can easily get the following proposition about $\W_1, \W_2.$

\begin{proposition}The matrix wave operators $\W_1, \W_2$ satisfy following Sato equations
\begin{align}
\label{Wjk} \partial_{t_{jk}}\W_1&=(B_{jk})_+\W_1,& \partial_{t_{jk}} \W_2&=(B_{jk})_+\W_2,  \\
\label{Wbjk}\partial_{\bar t_{jk}}\W_1&=-(\bar B_{jk})_- \W_1,& \partial_{\bar t_{jk}} \W_2&=-(\bar B_{jk})_- \W_2.  \end{align}
\end{proposition}

 From the previous proposition we can derive the following  Lax equations for the Lax operators.
\begin{proposition}\label{Lax}
 The  Lax equations of the MBTH(MCTH) are as follows
   \begin{align}
\label{laxtjk}
  \partial_{t_{jk}} \L&= [(B_{jk})_+,\L],&
 \partial_{t_{jk}} C_{ss}&= [(B_{jk})_+,C_{ss}],&\partial_{t_{jk}} \bar C_{ss}&= [(B_{jk})_+,\bar C_{ss}],
\\
  \partial_{\bar t_{jk}} \L&= [ (\bar B_{jk})_+,\L],&
 \partial_{\bar t_{jk}} C_{ss}&= [(\bar B_{jk})_+,C_{ss}],&\partial_{\bar t_{jk}} \bar C_{ss}&= [(\bar B_{jk})_+,\bar C_{ss}].
 \end{align}

\end{proposition}

\section{Symmetries of  of MBTH(MCTH)}

To introduce the additional symmetries of the MBTH(MCTH). The Orlov-Shulman operator of the MBTH(MCTH) will be defined
as
\begin{equation}\label{osoperator}
    \M_1=\W_1(\varepsilon\otimes E_{N\times N}) \W_1^{-1},\ \ \M_2=\W_2(\varepsilon^*\otimes E_{N\times N})\W_2^{-1},
\end{equation}
\begin{equation}\label{osoperator}
    R_{ij}=\W_i(E\otimes E_{jj}) \W_i^{-1}.
\end{equation}
To construct the additional quantum torus symmetry of the multicomponent BTH, firstly we define the operator  $B_{m nj}^{(i)}$
as
\begin{align}\label{defBoperator}
B_{m nj}^{(i)}=\M_i^m\L_i^{n}R_{ij}-(-1)^{n} R_{ij}\L_i^{n-1}\M_i^{m}\L_i.
\end{align}
For the multicomponent CTH, we define the operator  $B_{m nj}^{(i)}$ as
\begin{align}\label{defBoperator}
B_{m nj}^{(i)}=\M_i^m\L_i^{n}R_{ij}-(-1)^{n} R_{ij}\L_i^{n}\M_i^{m}.
\end{align}

For any  matrix operator $B_{m nj}^{(i)}$ in \eqref{defBoperator}, one has
\begin{align}\label{Bflow}
&\frac{\d B_{m nj}^{(i)}}{\d t_{kj}}=[(\L_1^kR_{1j})_+, B_{m nj}^{(i)}],  \ k\in\Zop.
\end{align}
\begin{align}\label{Bflow}
&\frac{\d B_{m nj}^{(i)}}{\d \bar t_{kj}}=[(\L_2^kR_{2j})_+, B_{m nj}^{(i)}],  \ k\in\Zop.
\end{align}
Then we can derive the following lemma.
\begin{lemma}\label{MBTHlemma}
The following identities hold true
\begin{eqnarray}
\bar \Lambda^{-1}(\varepsilon\otimes E_{N\times N}) \bar \Lambda=\J^{-1}(\varepsilon^T\otimes E_{N\times N}) \J,&&\bar \Lambda(\varepsilon^* \otimes E_{N\times N})\bar \Lambda^{-1}=\J^{-1}(\varepsilon^{*T}\otimes E_{N\times N}) \J,\label{blemma}
\end{eqnarray}
\begin{eqnarray}
\varepsilon(\varepsilon\otimes E_{N\times N})=\K(\varepsilon^T\otimes E_{N\times N}) \K^{-1},
&&\varepsilon^*\otimes E_{N\times N} =\K(\varepsilon^{*T}\otimes E_{N\times N}) \K^{-1}.\label{clemma}
\end{eqnarray}
\end{lemma}

Then for the MBTH, by (\ref{bcwconstraints}) and (\ref{blemma}), we can derive
\begin{eqnarray}
{\M_1}^T&=&({\W_1}(\varepsilon\otimes E_{N\times N}) \W_1^{-1})^T=( \W_1^{-1})^T(\varepsilon^T\otimes E_{N\times N}) \W_1^T\nonumber\\
&=&\J \W_1\J^{-1}(\varepsilon^T\otimes E_{N\times N})\J \W_1^{-1}\J^{-1}\nonumber\\
&=&\J \W_1\J^{-1}(\varepsilon^T\otimes E_{N\times N})\J \W_1^{-1}\J^{-1}\nonumber\\
&=&\J \W_1\bar \Lambda^{-1}(\varepsilon\otimes E_{N\times N})\bar  \Lambda \W_1^{-1}\J^{-1}\nonumber\\
&=&\J \W_1\bar \Lambda^{-1}(\varepsilon\otimes E_{N\times N})\bar \Lambda \W_1^{-1}\J^{-1}\nonumber\\
&=&\J \W_1\bar \Lambda^{-1} \W_1^{-1} \W_1\varepsilon \W_1^{-1} \W_1\bar \Lambda \W_1^{-1}\J^{-1}\nonumber\\
&=&\J{\L}_1^{-1}{\M}_1 {\L}_1\J^{-1},\label{bmtranspose}
\end{eqnarray}
Using the second equation in eq.\eqref{blemma}, we can also derive
\begin{eqnarray}
{\M_2}^T&=&\J{\L}_2^{-1}{\M}_2 {\L}_2\J^{-1}.
\end{eqnarray}
Similarly, for the CTH, we can derive
\begin{eqnarray}
{\M_i}^T&=&\K{\M}_i \K^{-1}.
\end{eqnarray}
Because of the constraints (\ref{bctlconstr}) on the Lax operators for
the MBTH(MCTH), we can  have
the following proposition.
\begin{proposition}\label{propaddisymmtodab3}
For the MBTH, it is sufficient to ask for
\begin{equation}\label{MBTHaddsymmconstr}
B_{m nj}^{(i)T}=-\J B_{m nj}^{(i)}\J^{-1},
\end{equation}
\end{proposition}
\begin{proof}

From (\ref{bctlconstr}) and (\ref{bmtranspose}), we have
\begin{eqnarray*}
(\M_i^m\L_i^{n}R_{ij})^T&=&R_{ij}^T(\L_i^{n})^T(\M_i^m)^T =(-1)^l\J{\L_i}^l\J^{-1}\J{\L_i}^{-1}{\M_i}^m{\L_i}\J^{-1}\\
&=&\J(-1)^{n} R_{ij}\L_i^{n-1}\M_i^{m}\L_i\J^{-1}.
\end{eqnarray*}
Since $J^T=J^{-1}=J$. Therefore
$B_{m nj}^{(i)}$ will satisfy the B type condition.
\end{proof}
Similarly, the following proposition can also be got.
\begin{proposition}\label{propaddisymmtodabC3}
For the MCTH, the following C type condition must hold true
\begin{equation}\label{MCTHaddsymmconstr}
B_{m nj}^{(i)T}=\K B_{m nj}^{(i)}\K^{-1},
\end{equation}
\end{proposition}
Now for the MBTH we will denote the  matrix operator $D_{m nj}$ as
\begin{equation}
D_{im nj}:=e^{m\M_i}q^{n\L_i}R_{ij}-\L_i^{-1}R_{ij}q^{-n\L_i}e^{m\M_i}\L_i,
\end{equation}
which further leads to
\begin{equation}
D_{im nj}=\sum_{p,s=0}^{\infty}\frac{m^p(n\log q)^s(\M_i^p\L_i^sR_{ij}-(-1)^sR_{ij}\L_i^{s-1}\M_i^p\L_i)}{p!s!}=\sum_{p,s=0}^{\infty}\frac{m^p(n\log q)^sB_{p sj}^{(i)}}{p!s!}.
\end{equation}
Then the following calculation will lead to the B(C) type anti-symmetry property of $D_{im nj}$ as
\begin{eqnarray*}D_{im nj}^T
&=&(\sum_{p,s=0}^{\infty}\frac{m^p(n\log q)^sB_{p sj}^{(i)}}{p!s!})^T\\
&=&-(\sum_{p,s=0}^{\infty}\frac{m^p(n\log q)^s\J B_{p sj}^{(i)}\J^{-1}}{p!s!})\\
&=&-\J(\sum_{p,s=0}^{\infty}\frac{m^p(n\log q)^sB_{p sj}^{(i)}}{p!s!})\J^{-1}\\
&=&-\J D_{im nj} \J^{-1}.
\end{eqnarray*}
Now for the MCTH we will denote the  matrix operator $D_{m nj}$ as
\begin{equation}
D_{im nj}:=e^{m\M_i}q^{n\L_i}R_{ij}-R_{ij}q^{-n\L_i}e^{m\M_i}.
\end{equation}
Therefore we get the following important B(C) type condition which the  matrix operator $D_{im nj}$ satisfies
\begin{equation}
D_{im nj}^T=-\J  D_{im nj} \J^{-1} (D_{im nj}^T=-\K  D_{im nj} \K^{-1}).
\end{equation}

Then basing on a quantum parameter $q$, the additional flows for the time variable $t_{m,n}^{ij},t_{m,n}^{*ij}$ are
defined as follows
\begin{equation}
\dfrac{\partial \S_1}{\partial t_{m,n}^{ij}}=-(B_{m nj}^{(i)})_-\S_1,\ \
 \dfrac{\partial \S_1}{\partial t^{*ij}_{m,n}}=-(D_{im nj})_-\S_1,
\end{equation}
\begin{equation}
\dfrac{\partial \S_2}{\partial t_{m,n}^{ij}}=(B_{m nj}^{(i)})_+\S_2,\ \
 \dfrac{\partial \S_2}{\partial t^{*ij}_{m,n}}=(D_{im nj})_+\S_2,
\end{equation}

or equivalently rewritten as

\begin{equation}
\dfrac{\partial \L_1}{\partial t_{m,n}^{ij}}=-[(B_{m nj}^{(i)})_-,\L_1], \qquad
\dfrac{\partial \M_1}{\partial t^{*ij}_{m,n}}=-[(D_{im nj})_-,\M_1],
\end{equation}

\begin{equation}
\dfrac{\partial \L_2}{\partial t_{m,n}^{ij}}=[(B_{m nj}^{(i)})_+,\L_2], \qquad
\dfrac{\partial \M_2}{\partial t^{*ij}_{m,n}}=[(D_{im nj})_+,\M_2].
\end{equation}

 Generally, one can also derive
\begin{equation}\label{bkpMLK}
\partial_{t^{*ip}_{l,k}}(D_{1m nj})=[-(D_{il kp})_-,D_{1m nj}],
\end{equation}

\begin{equation}\label{bkpMLK2}
\partial_{t^{*ip}_{l,k}}(D_{2m nj})=[(D_{il kp})_+,D_{2m nj}].
\end{equation}

 This further leads to the commutativity of the additional flow $\dfrac{\partial }{\partial
t^{*ij}_{m,n}}$  with the flow $\partial_{t_{jn}},\partial_{\bar t_{jn}}$ in the following theorem.

\begin{theorem}
The additional flows of $\partial_{t^{*is}_{l,k}}$ are  symmetries of the  multicomponent BTH(CTH), i.e. they commute with all $\partial_{t_{jn}},\partial_{\bar t_{jn}}$ flows of the   multicomponent BTH(CTH).
\end{theorem}
Comparing with the additional symmetry of the single-component BTH(CTH), the additional flows $\partial_{t_{l,k}^s}$ of the  multicomponent BTH(CTH) form the following $N$-folds direct product of the
$\W_{\infty}$ algebra  as following
\begin{eqnarray*}
&&[\partial_{t_{p,s}^{ir}},\partial_{t_{a,b}^{jc}}]\L_k=\delta_{ij}\delta_{rc}\sum_{\alpha\beta}C_{\alpha\beta}^{(ps)(ab)}
\partial_{t_{\alpha,\beta}^{ic}}\L_k,\ \ i,j,k=1,2; 1\leq r,c\leq N.
\end{eqnarray*}

Now it is time to identity the algebraic structure of the
additional $t_{l,k}^{*j}$ flows of the  multicomponent BTH(CTH).
\begin{theorem}\label{bkpalg}
The additional flows $\partial_{t^{*dj}_{l,k}}$ of the  multicomponent BTH(CTH) form the coupled $\bigotimes^NQT_+ $ algebra
( $N$-folds direct product of the
positive half of the quantum torus algebra $QT$), i.e.,
\begin{equation}
[\partial_{t^{*cr}_{n,m}},\partial_{t^{*dj}_{l,k}}]=\delta_{cd}\delta_{rj}(q^{ml}-q^{nk})\partial_{t^{*cr}_{n+l,m+k}},\ \ n,m,l,k\geq 0; \ \ 1\leq r,j\leq N;\ \ c=d=1,2.
\end{equation}

\end{theorem}

\begin{proof}

One can also prove this theorem as following by rewriting the quantum torus flow in terms of a combination of $\partial_{t_{m,n}^{ij}}$ flows
\begin{eqnarray*}
&&[\partial_{t^{*cr}_{n,m}},\partial_{t^{*dj}_{l,k}}]\L_i\\
&=&[\sum_{p,s=0}^{\infty}\frac{n^p(m\log q)^s}{p!s!}\partial_{t_{p,s}^{cr}},\sum_{a,b=0}^{\infty}\frac{l^a(k\log q)^b}{a!b!}\partial_{t_{a,b}^{dj}}]\L_i\\
&=&\sum_{p,s=0}^{\infty}\sum_{a,b=0}^{\infty}\frac{n^p(m\log q)^s}{p!s!}\frac{l^a(k\log q)^b}{a!b!}[\partial_{t_{p,s}^{cr}},\partial_{t_{a,b}^{dj}}]\L_i\\
&=&\sum_{p,s=0}^{\infty}\sum_{a,b=0}^{\infty}\frac{n^p(m\log q)^s}{p!s!}\frac{l^a(k\log q)^b}{a!b!}\sum_{\alpha\beta}C_{\alpha\beta}^{(ps)(ab)}\delta_{rj}\partial_{t_{\alpha,\beta}^{cr}}\L_i\\
&=&(q^{ml}-q^{nk})\sum_{\alpha,\beta=0}^{\infty}\frac{(n+l)^\alpha((m+k)\log q)^\beta}{\alpha!\beta!}\delta_{cd}\delta_{rj}\partial_{t_{\alpha,\beta}^{cr}}\L_i\\
&=&(q^{ml}-q^{nk})\delta_{cd}\delta_{rj}\partial_{t^{*cr}_{n+l,m+k}}\L_i.
\end{eqnarray*}
\end{proof}

%%%%%%%%%%%%%%%%%%%%%%%%%%%%%%%%%%%%%%%%%%%%%%%%%%%%%%

{\bf Acknowledgments:}
  { This work is supported by the National Natural Science Foundation of China under Grant No. 11571192 and K. C. Wong Magna Fund in
Ningbo University.}
%%%%%%%%%%%%%%%%% References  %%%%%%%%%%%%%%%%%%%%%%%%%%%%%%%%%%%%%%%

%%%%%%%%%%%%%%%%%%%%%%%%%%%%%%%%%%%%%%%%%%%%%%%%%%%%%%%%%%%


\begin{thebibliography}{99}
%%%%%%%%%%%%%%%%%%%%%%%%%%%%%%%%%%%%%%%%%%%%%%%%%%%%%%%%%%%%%%%%%%%%%%%

\bibitem{Toda}
M. Toda, Wave propagation in anharmonic lattices, J. Phys. Soc.
Jpn. 23(1967) 501-506.


 \bibitem{UT}K. Ueno, K. Takasaki, Toda lattice hierarchy,
In \emph{``Group representations and systems of differential
equations'' (Tokyo, 1982)}, 1-95, Adv. Stud. Pure Math. 4,
North-Holland, Amsterdam, 1984.

\bibitem{witten}
R. Dijkgraaf, E. Witten, Mean field theory, topological field
theory, and multimatrix models, Nucl. Phys. B 342(1990), 486-522.



\bibitem{cheng2011}J. P. Cheng, K. L. Tian and J. S. He.
The additional symmetries for the BTL and CTL hierarchies. J. Math.
Phys. 51 (2011) 053515.


 \bibitem{CDZ}
G. Carlet, B. Dubrovin, Y. Zhang,  The Extended Toda Hierarchy,
Mosc. Math. J.  4(2004), 313-332.




\bibitem{C}
G. Carlet, The extended bigraded Toda hierarchy, Journal of Physics A: Mathematical and Theoretical 39(2006), 9411-9435.



\bibitem{ourJMP}
 C. Z. Li, J. S. He, K. Wu, Y. Cheng,  Tau function and  Hirota bilinear equations for the extended  bigraded Toda
 Hierarchy, J. Math. Phys. 51(2010), 043514.



\bibitem{leurhirota}G. Carlet, J. van de Leur, Hirota equations for the extended bigraded Toda hierarchy and the total
  descendent potential of $\mathbb{P}^1$ orbifolds, Journal of Physics A: Mathematical and Theoretical, 46(2013), 405205-405220.




\bibitem{ourBlock}
 C. Z. Li, J. S. He, Y. C. Su, Block type symmetry of bigraded Toda hierarchy,
J. Math. Phys. 53(2012), 013517.

\bibitem{solutionBTH} C. Z. Li, Solutions of  bigraded Toda hierarchy, Journal of Physics A: Mathematical and Theoretical, 44, 255201(2011), arXiv:1011.4684.
\bibitem{RMP}C. Z. Li, J. S. He, Dispersionless bigraded Toda hierarchy and its additional symmetry, Reviews in Mathematical Physics 24(2012), 1230003.




 \bibitem{manasInverse2}   M. Ma\~{n}as, L. Mart\'{\i}nez Alonso, The multicomponent 2D Toda hierarchy: dispersionless limit, Inverse Problems, 25(2009), 11.
\bibitem{manasinverse} M. Ma\~{n}as, L. Mart\'{\i}nez Alonso, and C. \'{A}lvarez Fern\'{a}ndez, The multicomponent 2D Toda hierarchy: discrete
flows and string equations, Inverse Problems, 25(2009), 065007.

\bibitem{manasaInverse}C. \'{A}lvarez Fern\'{a}ndez, U. Fidalgo Prieto,  M. Ma\~{n}as, The multicomponent 2D Toda hierarchy: generalized matrix orthogonal polynomials, multiple orthogonal polynomials and Riemann--Hilbert problems, Inverse Problems, 26(2010), 055009.

 \bibitem{manas} C. \'{A}lvarez Fern\'{a}ndez, U. Fidalgo Prieto,  and M. Ma\~{n}as, Multiple orthogonal polynomials of mixed type:
Gauss-Borel factorization and the multi-component 2D Toda hierarchy, Advances in Mathematics, 227(2011), 1451-1525.

\bibitem{EMTH} C. Z. Li, J. S. He, On the extended multi-component Toda hierarchy, Math. Phys., Analyis and Geometry 17(2014), 377-407.



\bibitem{EZTH} C. Z. Li, J. S. He, The extended $Z_N$-Toda hierarchy,  Theor. Math. Phys. 185(2015), 1614-1635.



\end{thebibliography}
\end{document}